\documentclass{lms}
\usepackage{amsfonts}
\usepackage{amsmath}
\newtheorem{theorem}{Theorem}[section] % 1st argument is your name for it
\newtheorem{lemma}[theorem]{Lemma}     % 2nd argument is what is printed
\newtheorem{corollary}[theorem]{Corollary}
\newtheorem{proposition}[theorem]{Proposition}
\newnumbered{assertion}{Assertion}    % 1st argument is your name for it
\newnumbered{conjecture}{Conjecture}  % 2nd argument is what is printed
\newnumbered{definition}{Definition}
\newnumbered{hypothesis}{Hypothesis}
\newnumbered{remark}{Remark}
\newnumbered{note}{Note}
\newnumbered{observation}{Observation}
\newnumbered{problem}{Problem}
\newnumbered{question}{Question}
\newnumbered{algorithm}{Algorithm}
\newnumbered{example}{Example}
\newunnumbered{notation}{Notation}

\newcommand{\tr}{\, {\rm tr}\,}

\newcommand{\bid}{\, {\bf 1}}

\newcommand{\ZZ}{\mathbb{Z}}

\newcommand{\NN}{\mathbb{N}}

\newcommand{\cF}{\mathcal{F}}
\newcommand{\CC}{\mathbb{C}}
\newcommand{\bE}{\mathbb{E}}
\newcommand{\bP}{\mathbb{P}}

\def\nd0{n_{\Delta^\#}}

\def\Nd0{N_{\Delta^\#}}

\def\cd{C_d}
\DeclareMathOperator{\diam}{diam} 
\title[Lifshitz asymptotics for percolation Hamiltonians]% end with percent
 {Lifshitz asymptotics for percolation Hamiltonians} % This is the full title of the paper
\author{R. Samavat, P. Stollmann and I. Veseli\'c}

\classno{ 60K35 (primary), 35P15 (secondary)}
\extraline{Financial support by the DFG is gratefully acknowledged.}
\begin{document}
\maketitle

\begin{abstract}
We study a discrete Laplace operator $\Delta$ on percolation subgraphs of an
infinite graph. The ball volume is assumed to grow at most polynomially. We are interested in 
the behavior of the integrated  density of states near the lower spectral edge.
If the graph is a Cayley graph we prove that it exhibits  Lifshitz tails.
If we merely assume that the graph has an exhausting sequence with positive $\delta$-dimensional density,
we obtain an upper bound on the integrated  density of states of Lifshitz type.

\end{abstract}

%\part{Use this type of header for very long papers only}
% use lowercase except for proper names

\section{Introduction} % use lowercase except for proper names
\label{intro}

The present paper is a continuation of works
\cite{KirschM-06,MuSt07,AntunovicV-a08,AntunovicV-09} that deal with spectral asymptotics of a
particular class of random Hamiltonians, introduced in the  following way: Starting from an infinite graph one
considers subgraphs defined by a percolation process, i.e., by deleting sites or
bonds of the graph, respectively, by means of a random process. That leads to a
random family of Laplacians of these subgraphs.
%%%%%%%%%%%%%%%%%%%

The topic is influenced by two strains of previous development:
on the one hand the study of the distribution of spectrum near its minimum.
This takes oftentimes the form of an exponential asymptotics called \emph{Lifshitz tails}.
While this has been first analysed for random Schr\"odinger operators on continuum space, the 
first result concerned with Hamiltonians on the lattice $\ZZ^d$ is \cite{Simon-b85}, 
which established Lifshitz tails for the so-called \emph{Anderson model}.
If $N(E)$ denotes the integrated density of states at energy $E$ and the Anderson model is normalized
in such a way that zero is the almost sure bottom of the spectrum,  \cite{Simon-b85} states that 
% \begin{equation*}
$\lim\limits_{E\searrow 0} \frac{\ln |\ln N(E)|}{|\ln E|} = \frac{d}{2}$.
% \end{equation*}
The value $d/2$ is sometimes called Lifshitz exponent.
%%%%%%%%%%%%%%%%%%%

In the present paper we focus on a different, albeit related model, namely a \emph{percolation Hamiltonian}, 
sometimes also called \emph{quantum percolation problem}. It can be viewed as a hard-core Anderson model, where
``hard-core" refers to the fact that the potential barrier is impenetrable, i.e. infinitely high.
While classical percolation theory started in 1957 with the seminal papers \cite{BroadbentH-57,Hammersley-57} 
and has developed in one of the most studied and paradigmatic models of statistical mechanics, 
the quantum percolation problem was introduced only shorty later, in 1959, see \cite{deGennesLM-a59,deGennesLM-b59}.
Since then it has been treated both from the theoretical, e.g. in \cite{KirkpatrickE-72,ChayesCFST-86}, 
as well as from the computational physics point of view e.g. in \cite{KantelhardtB-02}.
It seems that the first papers that established basic properties of 
percolation Hamiltonians on the mathematical level of rigour were
\cite{Veselic-a05} and \cite{Veselic-b05}. There also the formulation of the quantum percolation 
problem on Cayley and quasi-transitive graphs was given.

The low energy spectral asymptotics for the quantum percolation problem 
on the integer lattice has been identified in the non-percolating phase in \cite{KirschM-06} 
and subsequently in the percolating phase in \cite{MuSt07}. 
These two papers combined establish a phase-transition 
manifested in the spectral properties of the graph Laplacian of the percolation clusters. 

In \cite{AntunovicV-09}, relying on \cite{AntunovicV-a08}, the results of \cite{KirschM-06} for the non-percolating regime
have been generalized to the setting of amenable Cayley graphs. 
In the present paper we extend these results to the
percolating phase, using an abstract method  \cite{St99}, which was already applied in  \cite{MuSt07}. 
If the volume of the balls in the Cayley graph grow polynomially with exponent $d$, we 
prove Lifshitz asymptotics with Lifshitz exponent $d/2$ as in the case of the $\ZZ^d$ lattice.
More generally, for (possibly) inhomogeneous graphs 
where the upper polynomial volume growth bound is complemented 
by the requirement that arbitrarily large portions of the graph have at least $\delta$-dimensional volume growth, 
we derive an upper bound on the integrated density of states which corresponds to the Lifshitz exponent $\delta/(d+1)$.	
The results hold equally for two types of random Hamiltonians:
the Anderson model and the percolation Hamiltonian.

These kinds of models are surveyed in the papers
\cite{AntunovicV-b08} and \cite{MuSt11}, where also the
relation to geometric group theory as well as the underlying stochastic model is
reviewed. As the reader will notice, our methods do not use any kind of geometric 
or stochastic translation invariance, so that the
upper bound we give holds in situations where neither the integrated density of
states is defined nor where the usual $0$-$1$ laws of percolation theory hold.

\section{Preliminaries}

%\subsection{Mathematical statements}

Let $G=(V,E)$ be a graph, consisting of countable set of vertices $V(G):=V$ and	
the set of edges $E(G):=E$, consisting of sets $\{x,y\} \subseteq V$. Since we
will assume all our graphs to be undirected and without loops and multiple
edges, it is possible and convenient to identify edges with two-element subsets of the set
of vertices in this way. A \emph{path} between $x$ and $y$ is a finite family of
edges connecting $x$ and $y$ in the obvious sense. A \emph{connected graph} is a
graph which has at least one path between each two vertices. The \emph{distance}
between two vertices $x$ and $y$ is the minimum length of paths between them and
is denoted by $d(x,y)$. The ball of radius $r$ at $x$ is a subset of $V$
denoted by $B_r(x)$ which contains all vertices with distance $r$ or less to
$x$. While later we will assume more restrictive conditions, we will now concentrate on the case that 
 \[
 \sup \{\deg_{G}(x) \mid x \in V  \} =:\cd<\infty ,
\] 
which implies, in particular, that the following basic object we look at, the graph Laplacian, is a bounded operator:
\begin{eqnarray*}
\Delta_G&:&\ell^2(V)\to\ell^2(V),\\
\Delta_G f(x)&:=& \sum_{y:\{ x,y\}\in E}(f(x)-f(y)).
\end{eqnarray*}
Note that, according to our sign convention, $\Delta_G$ defines a nonnegative bounded operator on $\ell^2(V)$.

We will be dealing with Laplacians on subgraphs that come in the following way; for a subset $V'\subseteq V$ the \emph{induced subgraph} is given by $G'=(V',E')$ where $E'=\{ e\in E\mid e\subseteq V'\}$ consists of those edges that join two vertices from $V'$.

For such an induced subgraph there are different Laplacians on $\ell^2(V')$ that have been considered: First of all $\Delta_{G'}=:\Delta_{V'}=:\Delta_{V'}^N$ is sometimes called \emph{Neumann Laplacian}. The name is plausible in that $\Delta_{V'}$ has decomposition properties analogous to those of the Laplacian with Neumann boundary conditions in euclidean space, see lemma \ref{lemma1.2} below. We define the \emph{Dirichlet Laplacian} by:

\begin{eqnarray*}
\Delta_{V'}^D&:&\ell^2(V')\to\ell^2(V'),\\
\Delta_{V'}^D f(x)&:=& \sum_{y:\{ x,y\}\in E'}(f(x)-f(y))+(\deg_{G}(x)-\deg_{G'}(x))f(x),
\end{eqnarray*}
so that 
$$
\Delta_{V'}^D = \Delta_{V'}^N + (D_G-D_{G'}) ,
$$
where $D_G$ and $D_{G'}$ denote the multiplication operators with the functions $\deg_{G}$ and $\deg_{G'}$, respectively. Note that the terminology is different at this point (it even differs from the one chosen in \cite{AntunovicV-09},\cite{MuSt07}). Our nomenclature seems to be justified by the following lemma whose straightforward proof we do not work out here. We need some notation:
Let  $Q$ be the form associated with $\Delta_G$, i.e.,
$$
Q(f,g):= (\Delta_G f\mid g)\mbox{  for  }f,g\in \ell^2(V) .
$$
Finally, the indicator function of the set $A$ is denoted by $1_A$.
\begin{lemma}\label{dirichlet}Let $V'\subseteq V$. 
\begin{itemize}
 \item[{\rm (1)}] $\Delta_{V'}^D$ is the unique selfadjoint operator in $\ell^2(V')$ associated to the restriction of $Q$ to $\ell^2(V')$, and the latter restriction in turn equals the closure of $Q$ restricted to $C_c(V')$ (continuity w.r.t. the discrete topology). 
 \item [{\rm (2)}]
$$\Delta_{V'}^D=\lim_n \big(\Delta_G+n1_{V\setminus V'}\big)
$$
in the generalized strong resolvent sense. 
  \item [{\rm (3)}] For the natural inclusion $J_{V'}:\ell^2(V')\to \ell^2(V)$
    $$\Delta_{V'}^D=J_{V'}^*\Delta_G J_{V'} .$$
    Note that $J_{V'}^*$ is nothing but restriction of functions.
\end{itemize}
\end{lemma}

\begin{remark}\begin{enumerate}
               \item In many papers, starting with \cite{Simon-b85}, the following operator
               $$\Delta_{V'}^{DD}:= \Delta_{V'}^N + 2(D_G-D_{G'})$$
               is called Dirichlet operator. Since it is bounded below by our Dirichlet Laplacian, the spectral estimates we obtain remain valid for $\Delta_{V'}^{DD}$ as well. The advantage
      of this definition lies in the following property:
      $$\Delta_{V_1 \cup V_2}^{DD} \le \Delta_{V_1}^{DD} \oplus \Delta_{V_2}^{DD}$$
      in the sense of quadratic forms for two disjoint $V_1, V_2 \subseteq V$.
      The similar inequality is not true for what we call the Dirichlet Laplacian due to the lack of locality of those operators.
       \item What we call Dirichlet Laplacian sometimes is called pseudo Dirichlet or Adjacency Laplacian or Laplacian with simple boundary conditions. Part (3) of the preceding
lemma means that our Dirichlet Laplacian can be thought of as just truncating the matrix of Laplacian 
of the ambient big graph.               
               \item For a thorough discussion of Dirichlet boundary conditions in the general Dirichlet form framework see \cite{St93}.
              \end{enumerate}
\end{remark}

%%%%%%%%%

The following lemma is again obviously true:

\begin{lemma}\label{lemma1.1} In the sense of quadratic forms we have 
$$0\leq \Delta_{V'}^N \leq \Delta_{V'}^D \leq \Delta_{V'}^{DD} \leq 2 D_G.$$
\end{lemma}

It is important to know that the Neumann Laplacians possess a
monotonicity property with respect to disjoint union of vertex sets. The
following lemma is easy to prove.

\begin{lemma}\label{lemma1.2} For two disjoint $V_1, V_2 \subseteq V$ we have
$$\Delta_{V_1}^N \oplus \Delta_{V_2}^N \leq \Delta_{V_1 \cup V_2}^N.$$
\end{lemma}
Next for given $E\in \mathbb{R}$ we define the eigenvalue counting function as
$$n(E;H)=\text{tr} 1_{(-\infty,E]}(H), $$
for any non-negative, selfadjoint operator $H$ on some Hilbert space. 
Typically, we will be
dealing with Laplacians on finite subgraphs so that the trace in question is
indeed finite. An evident consequence of the previous lemma is 
that the eigenvalue counting function for the Neumann Laplacian is
subadditive in the domain. This is formulated precisely in 
\begin{corollary}\label{cor1.3}
Let $V_1, V_2\subseteq V$ be disjoint. Then
$$ n(E;\Delta^N_{V_1 \cup V_2}) \leq  n(E;\Delta^N_{V_1}) + n(E;\Delta^N_{V_2}).$$
\end{corollary}

Now we define site percolation on the graph $G$. Fix $p \in [0,1]$ and for $x
\in V$ let $\omega_x$ be a Bernoulli random variable taking the value 1 with
probability $p$ and value 0 with probability $1-p$. We assume the $(\omega_x)_{x
\in V}$ to be independent and  denote the set of all possible configurations
$\omega =(\omega_x)_{x \in V}$  by $\Omega$ and the corresponding product
measure by $\mathbb{P}$ or $\mathbb{P}_p$, if we want to stress the dependence on the parameter $p$.
For a given configuration $\omega$ we call 
$$V_{\omega} = \{ x \in V | \omega_x =1 \}$$ 
the set of open sites, and define 
$$E_\omega = \{ \{x,y\}\in E(G) | x,y \in V_\omega \}.$$ 
The subgraph of $G_\omega=(V_\omega,E_\omega)$ is called a percolation subgraph.
It is the subgraph of $G$ induced by the subset $V_\omega$ of the vertex set which is formed by the open sites. 

We now introduce the central random objects of the present paper. Given a percolation process
$(\Omega, \bP)$ corresponding to $p\in [0,1]$ as above, we define
$$
\Delta^D_\omega := \Delta^D_{V_\omega} \mbox{  for  }\omega=(\omega_x)_{x\in
V}\in \Omega 
$$ 
This gives a random family of operators on the respective spaces
$\ell^2(V_\omega)$.

Moreover, we define a random Schr\"odinger operator on $\ell^2(V)$ by
$$
h_\omega= \Delta_{G} + W_\omega, %\label{random}
$$
where $W_\omega$ stands for multiplication with the function $1-\omega_x$ (appropriately restricted to the set of vertices considered, if necessary). 

\section{Upper spectral bounds for uniformly polynomially bounded graphs}
%%%%%%%%%%%%%%%%
This section contains the heart of our results. We show that the distribution of
low lying eigenvalues is very sparse with high probability for a percolation
subgraph of a given uniformly polynomially bounded graph $G=(V,E)$. The latter condition only enters later, see \eqref{polbound}, so for now we just keep the assumptions from the preceding section.

We will look at restrictions of the Dirichlet Laplacian to finite
subsets of the random graph. In fact, due to the fact that Dirichlet boundary conditions are stronger than the random potential, we may rather look at restrictions of $h_\omega$, where we use Neumann boundary conditions. Here are the
details: For $U\subseteq V$ define
\begin{equation}\label{definition-our operator}
h^U_\omega:=\Delta_U^N +W_\omega \bid_U .
\end{equation}

Since we are interested in the low lying spectrum, it is important to note that 
all types of Hamilton operators introduced so far are non-negative operators. For finite $G'$, the graph Laplacian
$\Delta_{G'}=\Delta_{G'}^N$ has zero in its spectrum and the multiplicity 
is equal to the number of connected components of $G'$.
For any finite rank operator $H$ denote by $E_1(H)$ its lowest eigenvalue.
The second eigenvalue $ \vartheta_U:=E_2(\Delta_U)$ of the graph Laplacian $\Delta_U$
of some $U$ finite subset of $V$ will play the role of a spectral gap in the sequel. It is positive if and only if $U$ is connected, a condition we keep for the sequel.
For the statement of our first theorem, 
recall that $p$ denotes the success probability of the Bernoulli random variables $\omega_x$.

\begin{theorem}\label{th1}
Fix $\alpha \in (0,1-p)$. For $E\leq \frac{\alpha^2}{42} \, \vartheta_U$:
\begin{equation}
 \label{proba}
\bP\{ E_1(h^U_\omega)\le E\} \le \exp [-\gamma | U|] .
\end{equation}
where 
$\gamma:= (1-p-\alpha)^2$. 
\end{theorem}

We will prove this theorem in two steps: an analytic and a probabilistic argument. 
The first one is purely deterministic 
and relies on an analytic perturbation method as developed in \cite{St99}.
For this purpose we define an auxiliary operator
$$h_{\omega}^U(t)=\Delta_U + t W_{\omega}
\quad \text{ for } t \in [0,1].
$$
Since $W_{\omega}$ is non-negative we have for all $t \in [0,1]$ 
\begin{equation}
\label{eq:ordering}
h_{\omega}^U(t) \leq  \Delta_U + W_{\omega} = h_\omega^U 
\end{equation}
Now the next lemma relates the position of the lowest eigenvalue $E_1(t,\omega):= E_1(h_{\omega}^U(t))$ to the derivative 
$E_1^{'}(t,\omega)=\frac{d}{dt}E_1(t,\omega)$ w.r.t. the parameter $t$ at $t=0$.
\begin{lemma}\label{lemmaalpha}
Fix $\alpha \in (0,1]$.
For $E\leq \frac{\alpha^2}{42} \, \vartheta_U$:
\begin{equation}
\label{probdif}
E_1(h^U_\omega)\le E \Longrightarrow E_1^{'}(0,\omega)\le \alpha
\end{equation}
\end{lemma}

\begin{proof}
We use holomorphic perturbation theory, using  $\vartheta_U>0$. Since $\omega$ is fixed, we mostly suppress it in the notation. Consider the real analytic operator function $h(t):=h_{\omega}^U(t)= \Delta_U + t W_{\omega}$. As $\|W_\omega\| \le 1$ we know that $E_1(\cdot):=E_1(\cdot,\omega)$ extends to a holomorphic function on the open disc $\{ z\in\CC\mid |z| < \frac12\vartheta_U\}$. We can use this to bound the second derivative $E_1''$ on the interval $[0,\frac18\vartheta_U]$ by
$$
\max_{s\in [0,\frac18\vartheta_U]}| E_1''(s)|\le \frac{16}{\vartheta_U},
$$
using the Cauchy integral formula with the a circle of radius $\frac14\vartheta_U$ as a contour. Since $E_1(0)=0$, Taylor's formula gives 
$$
| E_1(t)-tE_1'(0)|\le \frac{8}{\vartheta_U}t^2\mbox{  for all  }t\in [0,\frac18\vartheta_U] .
$$
Rearranging yields
$$
E_1'(0)\le \min_{t\in (0,\frac18\vartheta_U]} \big(\frac{8t}{\vartheta_U}+\frac{E_1(t)}{t}\big)
$$
and the inequality of arithmetic and geometric means gives us that the minimum should be $\sqrt{\frac{32E_1(t)}{\vartheta_U}}$. Now an application of Min-max-principle implies $E_1(t)\leq E_1(h_\omega^U)$ so
$$E_1'(0)\leq \sqrt{\frac{32E}{\vartheta_U}}\leq \alpha,$$
which gives the assertion
\end{proof}

The proof of theorem  \ref{th1} will be finished by probabilistic arguments.
The main tool for this will be a large deviations or concentration inequality in a particularly simple case:
\begin{remark}\label{hoeffding}Let $I$ be a finite set and $X_i,i\in I$ be i.i.d Bernoulli variables with success probability $q\in (0,1)$. 
Then, for $\alpha\in (0,q)$,
$$
\bP\{ \frac{1}{|I|}\sum_{i\in I}X_i\le \alpha\}\le \exp[-(q-\alpha)^2|I|] .
$$
\end{remark}

This is an immediate consequence of Hoeffding's bound, \cite{Hoeffding-63}; 
it could also be deduced from Chernoff's \cite{Chernoff-52} and was certainly known much earlier. Today one can even find it in Wikipedia :) 

\begin{proof}[of theorem~{\rm\ref{th1}}]
By the Feynman-Hellmann formula (\cite{St01} theorem 4.1.29), we have 
$$
E'_1(0,\omega)= (W_{\omega} \varphi_0 \mid \varphi_0), 
$$
where $\varphi_0=\dfrac{1}{|U|^{\frac{1}{2}}}$ is the normalized 
ground state of $\Delta_U$, so that
$$
\bP[E'_1(0,\omega) \leq \alpha] = \bP[\dfrac{1}{|U|} \sum_{x \in U}(1- \omega_x) \leq \alpha].
$$
Applying the previous remark as well as the previous lemma we get: 

For $E\leq \frac{\alpha^2}{42} \, \vartheta_U$:
$$
\bP[E_1(h^U_\omega)\le E]\le 
\bP[E'_1(0,\omega) \leq \alpha] \le \exp[-(1-p-\alpha)^2|U|] .
$$
\end{proof} 
Now we are heading towards the announced estimate for the spectral distribution function.
The above mentioned uniform polynomial bound
means that there exists constants $d\ge 0$ and $b_d > 0$  such that 
\begin{equation}\label{polbound}
 |B_r(x)| \leq b_d r^d \text{ for all $r \in \mathbb{N}$ and $x \in V$ }  
\end{equation}
Here $|.|$ is the \emph{volume} of the respective set, i.e., its cardinality.
In particular, the vertex degree is uniformly bounded
\[
 \sup \{\deg_{G}(x) \mid x \in V  \} =\cd \leq b_d -1.
\]

Moreover, we assume that $G$ is connected. This doesn't pose a real restriction: if the graph consists of several connected components, all operators we consider decompose and the results we have apply to the individual terms in the sum.

Note that
this setting is much too general to assure the existence of the integrated
density of states (IDS). (For a definition of the IDS in an appropriate setting,  
see \cite{Veselic-b05}, \cite{AntunovicV-b08}, or \cite{MuSt11}.)

Therefore we will work with quantities that are going
to agree with the IDS in case the latter exists. At the same time, our setting
is too general to allow for the usual $0$-$1$-laws of percolation theory, since there
is no translation invariance involved in our assumption. In nice spatially
homogenous graphs the IDS is the limit of eigenvalue counting functions of finite rank operators.

Let us now formalize the notion of integrated density of states or spectral distribution function, 
mentioned in the introduction, in a way adapted to our situation:
For a sequence $\cF=(F_j)_{j\in\NN}$  of finite subsets we call
\begin{equation}
\bar{N}_{\cF}(E,h):= \limsup \frac{1}{| F_j|}\bE\{
\tr\chi_{(-\infty,E]}(h_\omega^{F_j})\}\label{meanspectral}
\end{equation}
the \emph{mean spectral density bound along} $\cF$. The analogous quantity for the Dirichlet Laplacian is defined as follows. First we set
\begin{equation}
\label{eq:percolationHamiltonian}
H^F_\omega:= \Delta_F+(D_G-D_{G_\omega})\bid_F,
\end{equation}
which amounts to setting Dirichlet boundary conditions at the randomly removed vertices and Neumann boundary conditions at the other vertices in $V\setminus F$. Now,
$$
\bar{N}_{\cF}(E,H):= \limsup \frac{1}{| F_j|}\bE\{
\tr\chi_{(-\infty,E]}(H_\omega^{F_j})\} .
$$

If the graph $G$ is quasi-transitive and the sequence $\cF$ exhausts
$V$ in a suitable manner, the latter limit will agree with the IDS for the
Dirichlet Laplacian $\Delta^D_\omega$, see Section \ref{s:Cayley} below. Clearly, the former $\bar{N}_{\cF}(E,h)$ coincides with the IDS of the random Schr\"odinger operator $h_\omega$ and
$$
\bar{N}_{\cF}(E,H)\le \bar{N}_{\cF}(E,h)
$$ 
due to $h^F_\omega\le H^F_\omega$.

We say that $\cF$ is a sequence of $\delta$-\emph{dimensional density} $\eta >0$ if
\begin{equation}
\label{density}
\forall r\ge 0:\quad \frac{| \{ x\in F_n\mid |B_r^{F_n}(x)| 
\ge \eta  r^\delta\}|}{| F_n|}\to 1\mbox{  as  } n\to\infty .
 \end{equation}
Here $B_r^{F_n}(x)$ refers to the ball in the subgraph induced by $F_n$. 
Quite clearly, this condition means that most of the set is not too thin, 
compared with finite dimensional lattices. 
This condition is satisfied in many situations. 
In fact, for groups of polynomial growth it is a fundamental fact 
\cite{Gromov-81,vandenDriesW-84,Pansu-83} 
that balls with increasing radii have asymptotically matching upper and lower polynomial 
volume bounds. Consequently, they form a F\o{}lner sequence $(F_n)_n$ and 
satisfy conditions \eqref{polbound} and \eqref{density} with $\delta=d$.
\begin{remark} Note that condition \eqref{density} implies that $| F_n|\to\infty$ as $n\to\infty$. 
Moreover, \eqref{density} will hold with $\delta=1$ and $\eta=1$ if we take connected $F_n$ with $\diam (F_n)\to \infty$ as $n\to\infty$.
\end{remark}

\begin{theorem}
\label{thmspectralcount}
Assume that $G$ satisfies a uniform polynomial growth bound as 
in {\rm (\ref{polbound})} and let $\cF$  be a sequence of $\delta$-dimensional 
density $\eta$. Fix $\alpha \in (0,1-p)$. There is $\gamma(\alpha)>0$ such that for $E\leq \frac{\alpha^2}{42}$:
$$ 
\bar{N}_{\cF}(E,h)\le \exp[-\gamma(\alpha) E^{-\frac{\delta}{d+1}}] .
$$
\end{theorem}
%%%%
According to the assertion, the constant $\gamma(\alpha)$ may and will depend on the geometric properties. 
% An explicit possible choice of $\gamma(\alpha)$ is given in \eqref{} at the end of the proof. 
The proof will show that we can choose
$$
\gamma(\alpha)= \eta 2^{-\delta} \left(\frac{\alpha^2}{84b_d}\right)^{\frac{\delta}{d+1}}(1-p-\alpha)^2 .
$$
%%%%%%

Let us outline the main idea used in the proof of the above theorem. We can bound
$$
\bE[\tr\chi_{(-\infty,E]}(h_\omega^{F})]\le \bP[E_1(h^F_\omega)\le E]\cdot | F|$$
but we cannot apply our estimate \eqref{proba} from theorem \ref{th1} directly to the members of the sequence $\cF=(F_j)_{j\in\NN}$: since the $| F_j|\to\infty$ we will have an energy range (depending on $\vartheta_{F_n}\to 0$) that collapses to $E=0$. To avoid this complication, we decompose each $F=F_j$ into a collection of sets of approximately equal size and apply theorem \ref{th1} to each of these sets. Our first step is such a decomposition based on a Vorono\"{i} type construction:

\begin{proposition}\label{propo}
  Assume that $F\subseteq V$ is finite. For any $r\ge 0$ there exist $x_1, ... , x_m\in F$ and pairwise disjoint, connected subsets $V_1, ... ,V_m$ of $F$ so that
$$
B_{\frac{r}{2}}^F(x_k)\subseteq V_k\subseteq B_{r}^F(x_k)\mbox{  and  }\bigcup_{k=1}^m V_k=F .
$$  

\end{proposition}
\begin{proof}
Since $F$ is the disjoint union of its connected components it suffices to consider the case that $F$ is connected.
Consider $\mathcal{A}$, the set of all finite sets $N \subseteq F$ so that the different points of $N$ have mutual distance strictly larger than $r$ in $F$. 
We denote the distance function in $F$ by $d_F$. Obviously 
$$d_F(x,y) \geq d(x,y).$$   
Pick a maximal set $M$ from $\mathcal{A}$, such a maximal set can be obtained in a constructive manner by induction. 

Let $M=\{x_1,\ldots ,x_m\}$, so that $m=|M|$.

It is clear that $$\bigcup_{x \in M} B_r(x) = F,$$
otherwise, there would be $x_0 \in F$ with $d_F(x_0,x) > r$ for any $x\in M$ 
and so $M \cup \{ x_0\} $ would belong to $\mathcal{A}$, contradicting maximality. 

Now we define Vorono\"{i} cells in the following way: 
For $i \in \{1,...,m\}$ let
\begin{eqnarray*}
V_i&=&\{x \in F  | d_F(x_i,x) \leq d_F(x_j,x) \text{ for all } j > i\\  
&&\hspace{1.5cm}\text{and~~}d_F(x_i,x) < d_F(x_j,x) \text{ for all } j < i\}. 
\end{eqnarray*}
Clearly $B_{\frac{r}{2}}(x_i) \subseteq V_i$, since $d_F(x_i,x_j)>r$ for $i\neq j$. We show that $V_i$ is  star-shaped with center $x_i$. In fact, let $x \in V_i$ and choose a path in $F$ like $x_i=y_0,y_1,...,y_l=x$ with $l=d_F(x_i,x)$. It is easy to see that $\{y_1,...,y_l\}$ are in $V_i$, 
proceeding by induction from $y_{l-1}$ to $y_1$. 
If $y_{l-1} \notin V_i$ then $d_F(x_i,y_{l-1}) > d_F(x_j,y_{l-1})$ 
would hold for $j > i$ or $d_F(x_i,y_{l-1})\ge d_F(x_j,y_{l-1})$ for $j < i$. 
In the first case then
$$d_F(x_i,x) = d_F(x_i,y_{l-1}) + 1> d_F(x_j,y_{l-1})+1 \geq d_F(x_j,x),$$ 
a contradiction. The same can be done for the second case. Obviously the $V_i$ are disjoint and $V_k \subseteq B_r(x_k)$. To prove this statement suppose that there is $x\in V_k$ and $x\notin B_r(x_k)$ then because  $F \subseteq \bigcup_{i=1}^{m} B_r(x_k)$ we have $x\in B_r(x_{k'}) $ and then according to the definition of $V_i$ we have
$$d_F(x,x_k) \leq d_F(x,x_{k'})\leq r,$$
a contradiction. So $V_k \subseteq B_r(x_k)$.
\end{proof}

For any graph $G$ and any $U \subseteq V$, such that the induced subgraph is connected, 
the Cheeger bound 
\begin{equation}
\label{eq:Cheeger}
\vartheta_U \ge \frac{1}{| U| \diam (U) } , 
\end{equation}
holds true.
For  special geometries, e.g., for cubes in the usual integer lattice, 
or balls in Cayley graphs, much better estimates hold.
This will be exploited in Section   \ref{s:Cayley} below.

\begin{proof}[of theorem~{\rm\ref{thmspectralcount}}]
By the polynomial growth bound we have $|U| \diam (U)\leq 2b_d r^{d+1}$ 
for $U$ is contained in some ball of radius $r$. We use the Cheeger inequality to conclude that for $c(\alpha)=\frac{\alpha^2}{42}$ all
$$
 E \leq \frac{c(\alpha)}{2b_d r^{d+1}} 
$$
satisfy 
\begin{equation}
\label{eq:Cheeger1}
E 
%\leq \frac{c(\alpha,C_d)}{2b_d r^{d+1}} 
\leq \frac{c(\alpha)}{|U| \diam (U)} 
\leq c(\alpha) \, \vartheta_{U}
\end{equation}
thus a good choice for $r$ is 
\begin{equation}\label{Eandr}
r=\left(\frac{c(\alpha)}{2b_d}\right)^{\frac{1}{d+1}}E^{\frac{-1}{d+1}}.
\end{equation}
Now with $r$ as above we apply the previous lemma to $F_n$, 
and obtain $\{x_1, \dots,x_{m(n)}\}$ such that 
$$
d(x_k,x_l) > r   \text{~~for~~} k\neq l
$$
and disjoint, connected $V_1, \dots, V_{m(n)} \subseteq F_n$ such that 
$$
B_{\frac{r}{2}}^{F_n}(x_k) \subseteq V_k \subseteq B_{r}^{F_n}(x_k) .
$$
Note, first that the choice of $r$ implies
\begin{equation}
\label{eq:Cheeger-choice}
E\le c(\alpha)\vartheta_{V_k} ,
\end{equation}
so that we can apply theorem \ref{th1} to $V_k$. We next check that for most $k$ the volume of $V_k$ is not too small. In fact:
Order the  $x_k$ so that 
$|B_{\frac{r}{2}}^{F_n}(x_k)| \ge \eta  \left(\frac{r}{2}\right)^\delta $
for all $k \geq j(n)$.
By the density condition \eqref{density} we know that 
\begin{equation}\label{jn}
 \frac{j(n)}{| F_n|}\to 0\mbox{  for  }n\to\infty .
\end{equation}

As
$$
F_n=\bigcup_{k=1}^{m(n)}V_k
$$
is a disjoint union, lemma \ref{lemma1.2} gives that
$$
h^{F_n}_\omega\ge \bigoplus_{k=1}^{m(n)} h^{V_k}_\omega ,
$$
which in turn implies 
$$
\tr\chi_{(-\infty,E]}(h_\omega^{F_n})\le \sum_{k=1}^{m(n)}\tr\chi_{(-\infty,E]}(h_\omega^{V_k})
$$
The averaged normalized eigenvalue counting function
on $F_n$ then satisfies
\begin{eqnarray*}
\frac{1}{|F_n|} \bE\{\tr\chi_{(-\infty,E]}(h_\omega^{F_n})\}&\leq& \frac{1}{|F_n|} \sum_{k=1}^{m(n)}\bE\{\tr\chi_{(-\infty,E]}(h_\omega^{V_k})\}\\
&\leq& \frac{1}{|F_n|}\sum_{k=1}^{j(n)-1}\bE\{\tr\chi_{(-\infty,E]}(h_\omega^{V_k})\} \\
&& \hspace{0.5cm}+ \frac{1}{|F_n|}\sum_{k=j(n)}^{m(n)}\bE\{\tr\chi_{(-\infty,E]}(h_\omega^{V_k})\}.
\end{eqnarray*}
Using \eqref{jn} above, for each $\epsilon$ we can find $n$ large enough such that $\frac{j(n)-1}{|F_n|}\leq \epsilon$.
This yields
\begin{eqnarray*}
\limsup_{n \to \infty} \frac{1}{|F_n|}  \bE\{\tr\chi_{(-\infty,E]}(h_\omega^{F_n})\}\\
&&\hspace{-1cm} \leq \epsilon |V_k| + \lim_{n \to \infty} \frac{1}{|F_n|}\sum_{k=j(n)}^{m(n)}\bE\{\tr\chi_{(-\infty,E]}(h_\omega^{V_k})\}\\
&&\hspace{-1cm} \leq \epsilon |V_k| + \lim_{n \to \infty} \frac{1}{|F_n|}\sum_{k=j(n)}^{m(n)}|V_k|\bP\{ E_1(h^{V_k}_\omega)\le E\}.\\
&&\hspace{-1cm} \leq \epsilon |V_k| + \exp[-(1-p-\alpha)^2\nu]\\
\end{eqnarray*}
by theorem \ref{th1} (since \eqref{eq:Cheeger-choice} holds), where 
$$
\nu:=\min_{k=j(n), ..., m(n)} |V_k| .
$$
As $\epsilon$ was arbitrary, we get
\begin{equation}\label{limsup}
 \limsup_{n \to \infty} \frac{1}{|F_n|}  \bE\{\tr\chi_{(-\infty,E]}(h_\omega^{F_n})\} \le \exp[-(1-p-\alpha)^2\nu]
\end{equation}

We now estimate $\nu$ below in terms of $E$, making use of the $\delta$-dimensional density $\eta$ property of the sequence $(F_n)$ 
and the relation between $E$ and $r$, \eqref{Eandr} above:
\begin{eqnarray}
|V_k| &\geq& |B^{F_n}_{r / 2}(x_k)| \nonumber \\ 
&\geq& \eta\left(\frac{r}{2}\right)^\delta \nonumber \\ 
&=&\eta 2^{-\delta}\left(\frac{c(\alpha)}{2b_d}\right)^{\frac{\delta}{d+1}}E^{\frac{-\delta}{d+1}} . \label{eqn:poly}
\end{eqnarray}
Inserting this lower bound for $\nu$ into \eqref{limsup} above gives the assertion of the theorem with the choice
$$
\gamma(\alpha)= \eta 2^{-\delta} \left(\frac{\alpha^2}{84b_d}\right)^{\frac{\delta}{d+1}}(1-p-\alpha)^2 .
$$
\end{proof}
\section{Application to Cayley graphs}
\label{s:Cayley}
In this section we give an application of our above result to Cayley graphs, the actual starting point of 
the present paper. For more background we refer to \cite{AntunovicV-a08,AntunovicV-b08,AntunovicV-09}. Among the results there is the Lifshitz asymptotics for the Dirichlet Laplacian on percolation subgraphs of a given Cayley graph. Since we restrict ourselves to groups of polynomial growth, lower bounds can easily be derived by trial functions. The method of proof for the upper bound that was used in \cite{AntunovicV-09} only works in the subcritical phase, i.e., when $p$ is small enough, so that almost surely no infinite cluster occurs. That means that, almost surely, the Dirichlet Laplacian decomposes into a direct sum of finite rank operators and the upper bound is achieved by probability estimates on the cluster size. Our method does not feel the difference between the subcritical and supercritical phase; thus we can extend the upper bound for the integrated density of states for the whole range $p\in (0,1)$ in much the same way in which in \cite{MuSt07} the results from \cite{KirschM-06} had been 
extended for the lattice case.  
 
Suppose that $G$ is a finitely generated group and $S$ is a symmetric generating set not containing the identity element. 
We define the corresponding Cayley graph by $V(G):=G$ and $E(G):= \{\{g,gs\} \mid g\in G, s \in S\}$. 
This is a regular graph of degree $|S|$; note that $G$ stands  for the group and the graph at the same time, consistent with the common conventions.

Let us first comment on the consequences of a polynomial growth bound: 
Due to results of Bass \cite{Bass-72}, Gromov \cite{Gromov-81} and van den Dries and Wilkie \cite{vandenDriesW-84}, 
a polynomial bound like (3.5) above implies that $d\in\NN$, and that an analogous lower bound holds with the same exponent $d$. 
From now on, we suppose that $G$ is of polynomial growth.

Therefore, $G$ is also amenable. In fact, the sequence of balls with common center $(B_k)_{k\in\NN}$ forms a
F\o{}lner sequence. Passing to an appropriate subsequence $ (k_n)_{n\in\NN}$ of radii, 
one obtains even a tempered F\o{}lner sequence $\cF=(F_n)_{n\in\NN}$.

Due to ergodicity, cf.~\cite{Veselic-b05}, in that case
$$
\bar{N}_{\cF}(E,h)=\lim_n \frac{1}{| F_n|}\bE\{\tr\chi_{(-\infty,E]}(h_\omega^{F_n})\}
$$
and the Pastur-Shubin trace formula tells us that
$$
\bar{N}_{\cF}(E,h)=\bE\{\tr[\delta_g \chi_{(-\infty,E]}(h_\omega)\delta_g ]\}=:N(E,h),
$$
in particular, the limit is independent of the sequence. The same holds true for $\bar{N}_{\cF}(E,H)$ with the obvious notational changes. We get

\begin{theorem}\label{coro1}
Let $G$ be a Cayley graph of a finitely generated group of polynomial growth, with growth exponent $d$, and let $p\in (0,1)$. 
Then there are $\gamma, \tilde \gamma>0$ such that for $E$ small enough:
$$ 
\exp[-\tilde \gamma E^{-\frac{d}{2}}] \leq N(E,H)\le N(E,h)\le \exp[-\gamma E^{-\frac{d}{2}}] .
$$ 
\end{theorem}

To prove the upper bound, one notes that for Cayley graphs of polynomial growth
a better isoperimetric inequality holds than in general graphs, see e.g.~\cite{BabaiSzegedy-92} or \cite{Zu00}, namely
\[
\vert A\vert\leq (1+ \diam  (A)) \vert \partial A\vert\  .
\]
This implies, via Cheeger's constant, cf.~\cite{Ch97}, a lower bound on the spectral gap, $\vartheta_{V}$. 
It reads 
$$
\vartheta_{V}\ge\frac{const.}{(1+\diam  V)^2} 
$$
where the constant depends only on the group and the chosen set of generators.
Using this in the proof of theorem 3.5 (instead of inequality (\ref{eq:Cheeger})) above, we get the asserted upper bound. 
The lower bound refers to the integrated density of states of the percolation Hamiltonian 
defined in \eqref{eq:percolationHamiltonian}. In this setting, the lower bound holds for all values of $p\in (0,1)$:
there is no difference between the percolating and the non-percolating phase, cf.~theorem 13 and Remark 15 in \cite{AntunovicV-b08}.

The above theorem implies a Lifshitz-type double logarithm asymptotics in the sense that 
\begin{equation*}
\lim_{E\searrow 0} \frac{\ln |\ln N(E,H)|}{|\ln E|} = \frac{d}{2} \text{ and }
\lim_{E\searrow 0} \frac{\ln |\ln N(E,h)|}{|\ln E|} = \frac{d}{2}.
\end{equation*}
For the existence and more details on critical probabilities for Cayley graphs, we refer to \cite{AntunovicV-a08}.
\\[1em]
\begin{acknowledgements}\label{ackref}
 The authors are grateful to Francisco Hoecker-Escuti for useful discussions.
\end{acknowledgements}

\affiliationone{% in this example, two authors share an institution
   R. Samavat, P. Stollmann and I. Veseli\'c\\
   Faculty of Mathematics\\
   Technische Universität Chemnitz\\
   09107 Chemnitz\\
   Germany
   \email{reza.samavat@mathematik.tu-chemnitz.de\\
   peter.stollmann@mathematik.tu-chemnitz.de}
   }
% Important: Do not put any empty line here.
\end{document}